\documentclass[a4,12pt,twoside]{article}

\usepackage{pstricks,graphicx,eucal,mathrsfs}
\usepackage{amssymb,amsmath,amsxtra,eurosym}
\usepackage[utf8]{inputenc}
\DeclareUnicodeCharacter{20AC}{\euro}

\newtheorem{theorem}{Theorem}[section]

\newtheorem{lemma}[theorem]{Lemma}
\newtheorem{definition}[theorem]{Definition}

\newtheorem{corollary}[theorem]{Corollary}
\newenvironment{proof}[1][Proof]{\noindent\textbf{#1.} }{\ \rule{0.5em}{0.5em}}

\newcommand{\R}{\mathbb{R}}
\newcommand{\N}{\mathbb{N}}

\newcommand{\E}{\mathbb{E}}
\newcommand{\F}{\mathscr{F}}

\begin{document}
\markboth{{\sc Gottschalk, Nizami, Schubert}}{{\sc Options in Markets with unknown dynamics}}

\title{{ \Large \sc Option Pricing in Markets with Unknown Stochastic Dynamics}}
\author{{\sc Hanno Gottschalk}, {\sc Elpida Nizami} and {\sc Marius Schubert} \\
{\small Fachgruppe f\"ur Mathematik und Informatik, Bergische Universit\"at Wuppertal, Germany}\\{\tt \small $\{$hanno.gottschalk,elpida.nizami,marius.schubert$\}$@uni-wuppertal.de}} 


\maketitle

{\abstract \noindent } We consider arbitrage free valuation of European options in Black-Scholes  and Merton markets, where the general structure of the market is known, however the specific parameters are not known. In order to reflect this subjective uncertainty of a market participant, we follow a Bayesian approach to option pricing. Here we use historic discrete or continuous observations of the market to set up posterior distributions for the future market. Given a subjective physical measure for the market dynamics, we derive the existence of arbitrage free pricing rules by constructing subjective option pricing measures. The non-uniqueness of such measures can be proven using the freedom of choice of prior distributions. The subjective market measure thus turns out to model an incomplete market. In addition, for the Black-Scholes market we prove that in the high frequency limit (or the long time limit) of observations, Bayesian option prices converge to the standard BS-Option price with the true volatility. In contrast to this, in the Merton market with normally distributed jumps Bayesian prices do not converge to standard Merton prices with the true parameters, as only a finite number of jump events can be observed in finite time. However, we prove that this convergence holds true in the limit of long observation times.   

\vspace{.3cm}


\section{Introduction}

We consider the thought experiment, where a market participant prices options based on a market model $S_t$ for the underlying asset where the market's general structure is known as, however the  model parameters are not known. The models considered here are the Black-Scholes \cite{BS,Mer1} (pure exponential diffusion) and Merton \cite{Mer2} (exponential jump diffusion) models.  The calibration of the parameters is based on discrete (low frequency) or continuous (high frequency) observations of $S_t$ in the time preceding interval $[-\tau,0]$, where $\tau>0$ is the observation time. So we deal with historical volatility for the Black Scholes model and historical jump frequency and height distribution for the Merton model. Furthermore, the market participant follows an Bayesian approach in order to express her or his uncertainty about the parameters of the underlying model $S_t$. Hence, the subjective market measure $P(A)=\int P_\theta f(\theta) d\theta$ for the market is a mixture of the parametric family of market measures with respect to the model parameters $\theta$ and with weights given by a posterior distribution $f(\theta)$.

The main results of this paper are the following: First, we prove that for the Black-Scholes and Merton markets arbitrage free pricing rules (or equivalent martingale measures) $Q=\int Q_\theta f(\theta) d\theta$ exist and can be obtained as a mixture from a parameter-dependent family of equivalent martingale measures $Q_\theta$.   As $f(\theta)$ in the definition of $Q$ does not necessarily have to be the same posterior distribution as in the definition of $P$, while $Q$ remains to be an equivalent martingale measure with respect to $P$ as long as both posterior distributions are equivalent, we conclude that subjective markets are incomplete even in the case where the corresponding market for a fixed set of the parameters is a complete market, as in the case of the Black-Scholes market. Although this result seems to be rather natural, a proof seems to be missing in the literature. 

Second, we prove that the option prices for European options obtained by the pricing measure $Q$ converge almost surely to the Black-Scholes prices in the limit of high frequency observations or long time observations at a given frequency. We give a complete proof for this result for normalizable and non normalizable prior distributions using a saddle point argument, which is a variant  of Bayesian consistency \cite{CR,Gho}. However, in a Merton market where $Q$ is obtained as a posterior mixture of, e.g.,  mean corrected equivalent martingale measures $Q_\theta$, the $Q$-prices do not converge in the high frequency limit. The result would be the same for any other construction of $Q_\theta$, e.g. by the Esscher transform \cite{CT} and we choose mean correction only for convenience. In the limit of long time observations, $Q$ prices however converge almost surely to the mean corrected prices with respect to $Q_{\theta_0}$, where $\theta_0$ is the 'true' set of parameters.  The reason that prices after a finite observation time $\infty>\tau>0$ remain different from the standard Merton prices lies in the fact that (almost surely) only a finite number of jumps can be observed in finite time. Consequently, the subjective uncertainty abaout the true distribution and frequency of the jumps does not go away after a finite observation of the market $S_t$.  This implies that Bayesian option pricing is somewhat inconsistent in the case of the Black-Scholes market since the outcome depends very much on the observation frequency, this is not the case for markets of jump-diffusion type, like the Merton market. The fact that in the long time asymptotics also the $Q$-prices converge to $Q_{\theta_0}$ prices is less relevant, since the statistical law of empirical market data typically changes significantly during a few years -- a time span, where only typically a hand full of major jump events is observable.

The Bayesian approach to option pricing has been considered before in various publications, see e.g. \cite{DSa,GHM,RHBF} for an review of the early literature. In \cite{GHM} results are obtained that are  close to our findings on the Black-Scholes market in Section 2. However, the approach is somewhat reversed as the subjective market measure is derived from the subjective pricing measure, while we proceed the other way round. A clear statement on the equivalent martingale property is missing, although the paper contains some observations that go in that direction. The paper \cite{DSa,HLM} are in a similar setting, however the focus is on numerics and applications and not on the underlying mathematical structure. The paper \cite{FGMM} applies option pricing in the Baesian Black-Scholes market to real maket data.  The paper however contains a proof of consistency in the Black-Scholes case. 

In \cite{FMW}, a stochastic volatility model is trated in the Bayesian framework, but the focus is on a filtering technique for the stochastic volatility. Also \cite{Kai} discusses stochastic volatility (Heston) models in the Bayesian framework.   In \cite{RS}, the Baesian risk neutral dynamics is considered in a time series framework,  with GARCH models being the main focus. Also, the work \cite{JP} follows a similar approach, but also contains applications to portfolio management. The more general case of jump-diffusions is not treated in any of these papers, see however \cite{FR} for a recent numerical study. As explained previously, the jump diffusion case is of independent conceptual interest, especially in the context of high frequency observations.

The paper is organized as follows: In Section 2  we introduce the subjective Black Scholes market and pricing measures an prove equivalence and the martingale property in Theorem \ref{theo:NoArbitrageBS}.  We also give a prove of convergence of $Q$-prices to the usual Black Scholes prices in the high frequency (and long observation time) limit of observations in Theorem \ref{theo:BayesConsistency} for the convenience of the reader, reproducing essentially prior findings from \cite{GHM}. We also provide a numerical convergence study which shows that the usual 20-200 day-to-day  estimates of historical volatility do have sufficiently Bayesian uncertainty left such that Bayesian prices still are significantly different from standard BS prices. This underlines the importance of intra day quotes for the eliminition of Bayesian uncertainty in the BS-case. 

   Section 3 deals with the subjective Merton market for high frequency (continuous) observations such that the BS-part of the Merton model is fixed from the observation of a small piece of the trajectory. This is however not true for the jump part. We construct the posterior distribution $Q$ from continuous observation of the market via a Grisanov-like theorem for compound Poisson processes \cite[Chapter 10.5]{CT}.   The convergence of the subjective Merton prices to the mean corrected Martingale measure $Q_{\theta_0}$ is proven for the limit of long observation time. Some technical details can be found in the appendix. While  this is quite similar to the BS-situation mathematically, the main economic difference lies in the fact that it is not possible to generate more information from a higher observational frequency as jump events remain to be sparse. This is also illustrated by a numerical example that reveals considerably higher Bayesian Merton prices than Merton prices without Bayesian uncertainty even after an observation time of two years. 

 In the final section we give our conclusions. In order to keep the paper self consistent, our formulation of the saddle point method for Bayesian consistency is given in Appendix A.

\section{The Black Scholes Market with Unknown Volatility}
\subsection{Some Fundamentals on the Black Scholes Market}
We first collect some well-known facts on the Black-Scholes (BS) model \cite{BS,Mer1}. We thus consider an asset with price $S_t$ given by an exponential Brownian motion
\begin{equation}
\label{eqa:BSmarket}
S_t=S_0e^ {X_t}\mbox{ with } X_t=\rho t+\sigma W_t,~~t\in[-\tau,T].
\end{equation}
Here $W_t$ is a standard Brownian motion, that is conditioned to zero at $t=0$ and runs backward in time for $-\tau\leq t\leq 0$. $\rho$ is some interest rate which is assumed to be known, e.g. as the libor interest rate. $\sigma>0$ is the volatility which either has to be calculated -- as it is implicitly contained in public option price date -- or has to be estimated statistically from historic date. Here we follow the latter approach. $\tau >0$ is the time in the past for which a market participant assumes that the market dynamics has not changed significantly. The present time is $t=0$. $T$ is the maturity time of some option that we are going to consider.

Let $(\Omega,(\mathscr{F}_t)_{t\in[-\tau,T]},P_\sigma)$ be a filtered probability space such that the usual conditions are fulfilled and $S_t$ and $X_t$ are realized as adapted processes. Let $(\F_t^+)_{t\in[0,T]}$ be a second filtration such that $\F_t^+\subseteq \F_t$ for $0\leq t \leq T$, $(S_t)_{t\in[0,T]}$ and $(S_t)_{t\in[0,T]}$ are adapted processes with respect to $(\F_t^ +)_{t\in[0,t]}$ and the increments of $X_t$ in the past, $X_s-X_t$, $-\tau\leq s<t\leq 0$ generate a sigma algebra, such that $\F_T^+$ is independent from it under $P_\sigma$.

  It is well known from the fundamental theorem of option pricing \cite{DS,Schach} or \cite[Proposition 9.2]{CT} that an arbitrage fee price for contingent claims with $\mathscr{F}_T^+$- measurable, non-negative pay off $H:\Omega\to\R_+$ at maturity time $T$ is given at the present time $t=0$ by
\begin{equation}
\label{eqa:MartingaleMeasure}
V_0(H,\sigma)=e^{-\rho T}\E_{Q_\sigma}[H],
\end{equation} 
where $Q_\sigma$ is a measure which is equivalent to $P_\sigma$ such that $\hat S_t=e^{-\rho t}S_t$, $t\in[0,T]$ is a (local) $Q_\sigma$ martingale with respect to $\mathscr{F}_{t\in[0,T]}$. Furthermore, the market is complete, if and only if $Q_\sigma$ is uniquely determined by the martingale condition \cite[Chapter 9.2]{CT}.

For the BS-market with volatility $\sigma$ the equivalent martingale measure can be constructed by  Grisanov's formula
\begin{equation}
\label{eqa:Grisanov}
Q_\sigma =L_TP_\sigma\mbox{ with } L_T=e^{-\frac{1}{2}\sigma W_T-\frac{1}{8}\sigma^ 2T}.
\end{equation} 
Furthermore, as the BS market is complete, $Q_\sigma$ is unique.

We consider the European call and put options with maturity $T>0$ and strike price $k>0$, $C(K,T)$ and $P(K,T)$, that are defined by the pay off function $(\pm (S_T-K))^+$. The well-known BS-formula then provides the fair prices for these options
\begin{align}
\begin{split}
\label{eqa:BSprices}
V_0(C(T,K)|\sigma)&=S_0\Phi(d_1(\sigma))-e^{-\rho T}K\Phi(d_2(\sigma))\\
V_0(P(T,K)|\sigma)&=e^{-\rho T}K\Phi(-d_2(\sigma))-S_0\Phi(-d_1(\sigma))
\end{split}
\end{align}
with $d_{1/2}(\sigma)=[\log\left(\frac{S_0}{K}\right)+(\rho\pm \frac{1}{2}\sigma^ 2)]/(\sigma\sqrt{T})$. It is immediate that the right hand side is non-negative, is bounded by $S_0$ and $K$, respectively and depends continuously on the volatility $\sigma$.

\subsection{Pricing with an unknown volatility} 
For the remainder of this section we make the assumption that a market participant wants to price European options at time $t=0$ on the basis of observations $S_{t_1},\ldots,S_{t_{n+1}}$ of (positive) market quota at observation times $-\tau=t_1<t_2<\cdots<t_{n+1}=0$. As a definite value for the volatility is not fixed by this finite set of observations, she/he follows a Bayesian approach using that under $P_\sigma$
\begin{equation}
X_j=\frac{\log\left(\frac{S_{t_{j+1}}}{S_{t_j}}\right) -\rho \Delta t_j}{\sqrt{\Delta t_j}}\sim N(0,\sigma^2),~~\Delta t_j=t_{j+1}-t_j,
\end{equation}
holds for $j=1,\ldots,n$, where $N(0,\sigma^2)$ stands for the normal distribution with zero mean and  variance $\sigma^2$. Let $\pi(\sigma^2)\geq 0$ be some prior function of $\sigma^2)$, which we assume to be continuous and bounded. Let $\hat \sigma^2_n=\frac{1}{n}\sum_{j=1}^nX_j^2>0$, then the well-known a posteriori distribution for the variance $\sigma^2$ for Gaussian data is well defined for $n\geq 2$
\begin{equation}
\label{eqa:postSigma}
f_n(\sigma^ 2)=f_n(\sigma^ 2|\pi)=\frac{\frac{1}{\sigma^ n}e^ {-\frac{n}{2}\left(\frac{\hat\sigma_n^ 2}{\sigma^ 2}\right)}\pi(\sigma^ 2)}{\int_{\mathbb{R}_+} \frac{1}{s^ n}e^ {-\frac{n}{2}\left(\frac{\hat\sigma_n^ 2}{s^ 2}\right)}\pi(s^ 2)\,ds^ 2}.
\end{equation}
Here we suppressed the dependence of $f_n(\sigma^2)$ of the actually observed values $S_{t_j}$ and of the prior $\pi(\sigma^2)$ for notational simplicity. Note that the denominator for the non informative prior $\pi(\sigma^2)= 1$ is proportional to the likelihood, given $\hat\sigma^2_n$.

From a Bayesian standpoint, the following definition is natural:
\begin{definition}[Subjective BS Market Model]
\label{def:subjBS}
Suppose that for $\sigma^2>0$, $P_\sigma$ is a family of measures on $(\Omega,(\mathscr{F}_t)_{t\in[-\tau,T]})$ such that $S_t$ under $P_\sigma$ is distributed as in \eqref{eqa:BSmarket} and fulfils the conditions given above.

 Furthermore, for $A\in\mathscr{F}_T^+$, $\sigma^ 2\mapsto P_\sigma(A)$ is Borel measurable in $\sigma^ 2$ on $\R_+$. Let $S_{t_1},\ldots,S_{t_{n+1}}$ be the observations available from the past and $f_n(\sigma^2)$ the a posteriori distribution associated with some prior $\pi(\sigma^2)$, see \eqref{eqa:postSigma}. Then
\begin{equation}
\label{eqa:subjMarketMeasure}
P_\pi(A)=\int_{\R_+} P_\sigma(A) f_n(\sigma^2)\, d\sigma^2,~~~A\in\mathscr{F}^+_T,
\end{equation}
defines a probability measure that we call the subjective market measure for the BS-market (given the observations $S_{t_1},\ldots,S_{t_{n+1}}$ of the past).

Furthermore, define the subjective BS pricing measure
\begin{equation}
\label{eqa:subjPricingMeasure}
Q_\pi(A)=\int_{\R_+} Q_\sigma(A) f_n(\sigma^2)\, d\sigma^2,~~~A\in\mathscr{F}^+_T.
\end{equation}
For the non informative prior $\pi(\sigma^ 2)=1$, we also write $P=P_1$ and $Q=Q_1$.
\end{definition} 
Nomalization $P_\pi(\Omega)=1$ ($Q_\pi(\Omega)=1$) follows from normalization of $P_\sigma$ ($Q_\sigma$) and $f_n(\sigma^2)$ and sigma-additivity is an easy consequence of the sigma additivity of $P_\sigma$ ($Q_\sigma$) and monotone convergence for the $d\sigma^2$-Lebesgue integral. 

Note that mesurability of $P_\sigma(A)$ ($Q_\sigma(A)$) in $\sigma$ can be verified with the aid of the following construction: Let $P_B$ be the canonical measure on the continuous function $(C(\R_+),\mathscr{B}(C(\R_+))$ endowed with the Borel sigma algebra. Let $f(\sigma^2)$ be some positive, measurable function on $(\R_+,\mathscr{B}(\R_+))$ and let $\varphi:C(\R_+)\times \R_+\to C(\R_+)$ given by $(\omega(\cdot),\sigma^2)\mapsto \omega(\cdot/\sigma)$. Then, the image maeasure of $P_B\otimes f(\sigma^2)d\sigma^2$ under the mapping $\varphi$ is a construction of $P$. The existence of measurable kernels $P_\sigma(A)$ for $B\in\mathscr{B}(C(\R_+))$ now follows from Fubini's theorem \cite{Hal} applied to the product measure $P_B\otimes f(\sigma^2)d\sigma^2$.

\begin{theorem}[Arbitrage Free Pricing for the Subjective BS Market] 
\label{theo:NoArbitrageBS}
Let $\pi(\sigma^2)$ and $\pi'(\sigma^2)$ be two functions on $\R_+$ such that $\pi(\sigma^2)d\sigma^2$ and $\pi'(\sigma^2)d\sigma^2$ are equivalent. Then $Q_{\pi'}$ is an equivalent martingale measure with respect to $P_\pi$.

Furthermore, the subjective BS-market defined by $P_\pi$ is incomplete.  

\end{theorem} 

\begin{proof}
The second assertion is an easy consequence of the first and the equivalence between uniqueness of the martingale measure and market completeness by the second fundamental theorem of asset pricing \cite[Proposition 9.3]{CT}. Note that different choices of $\pi'(\sigma^2)$ lead to different measures $Q_{\pi'}$.

For the equivalence of $P_\pi$ and $Q_{\pi'}$ let $A\in \F_T^+$ be a $P_\pi$ null set. Then, for $f_n(\sigma^2|\pi)d\sigma^2$ almost all $\sigma^2$ we have $P_{\sigma^2}(A)=0$ since otherwise the $\sigma^2$ integral would be positive. By equivalence of $P_\sigma$ and $Q_\sigma$ this implies $Q_\sigma(A)=0$ holds $f_n(\sigma^2|\pi)d\sigma^2$ almost surely and thus $f_n(\sigma^2|\pi')d\sigma^2$ almost surely since these two measures on $(\R_+,\mathscr{B}(\R_+))$ are equivalent with Radon-Nikodyn derivative given up to a positive constant by $\frac{\pi(\sigma^2)}{\pi'(\sigma^2)}$. Now, $Q_{\pi'}(A)$ vanishes as an integral over an $F_n(\sigma^2|\pi')d\sigma^2$ almost surely vanishing function in $\sigma^2$. Interchanging the r\^ole of $P_\pi$ and $Q_{\pi'}$ in the above argument, we derived equivalence.

To show the martingale property of $\hat S_t$ under $Q_{\pi'}$, we choose $A\in\F^+_s$ and let $-\tau\leq s<t\leq T$. Then, by the fact that $\hat S_t$ has the martingale property under all $Q_\sigma$, we obtain
\begin{align}
\begin{split}
\E_{Q_{\pi'}}[1_A\E_{Q_{\pi'}}[\hat S_t|\F^+_s]]&=\E_{Q_{\pi'}}[1_A\hat S_t]\\
&=\int_{\R_+}\E_{Q_\sigma}[1_A\hat S_t]\,f_n(\sigma^ 2)\,d\sigma^ 2\\
&=\int_{\R_+}\E_{Q_\sigma}[1_A\E_{Q_\sigma}[\hat S_t|\F^+_s]]\,f_n(\sigma^ 2)\,d\sigma^ 2\\
&=\int_{\R_+}\E_{Q_\sigma}[1_A\hat S_s]\,f_n(\sigma^ 2)\,d\sigma^ 2=\E_{Q_{\pi'}}[1_A\hat S_s].
\end{split}
\end{align}
As $A\in \F^+_s$ is arbitrary and $\hat S_s$ and  $\E_{Q_{\pi'}}[\hat S_t|\F^+_s]$ are both $\F^+_s$-measurable, it follows that   $\hat S_s=\E_{Q_{\pi'}}[\hat S_t|\F^+_s]$ $Q_{\pi'}$-a.s., which is the martingale property.
\end{proof}

From the theorem and \eqref{eqa:BSprices} one now deduces:
 \begin{corollary}[Subjective BS Option Prices] The arbitrage free prices with respect to the martingale measure $Q_\pi$ are given by 
 \begin{align}
\begin{split}
\label{eqa:BSpricesSubjective}
V_0(C(T,K)|\pi,n)&=\int_{\R_+}[S_0\Phi(d_1(\sigma))-e^{-\rho T}K\Phi(d_2(\sigma))]\, f_n(\sigma^ 2|\pi)\, d\sigma^ 2,\\
V_0(P(T,K)|\pi,n)&=\int_{\R_+}[e^{-\rho T}K\Phi(-d_2(\sigma))-S_0\Phi(-d_1(\sigma))]\, f_n(\sigma^ 2|\pi)\, d\sigma^ 2,
\end{split}
\end{align}
where we again suppressed the dependence on the past observations. 
 \end{corollary}

\subsection{The limit of high frequency or long time observations}

Here we consider the limit when the number of observations in the past goes to infinity and the market dynamics follows \eqref{eqa:BSmarket} for some fixed $\sigma_0>0$, which is however unknown to a market participant. Let $P_\sigma$ be the associated market measure. The limit $n\to\infty$ of the number of observations going to infinity can be realized either by letting $\tau\to\infty$ and keeping the frequency of observations fixed, or by increasing the frequency of observations keeping $\tau$ fixed. Technically, this problem falls into the field of Bayesian consistency, see e.g. \cite{CR,Gho}. We prove:

\begin{theorem}[Convergence of Option Prices to Standard BS Prices]
\label{theo:BayesConsistency}
In the limit when the number of past observations $n$ goes to infinity, the subjective BS-prices for European options converge to the BS-prices with volatility $\sigma_0$, provided $\pi(\sigma_0^2)>0$. We have
 \begin{align}
\begin{split}
\label{eqa:BSpricesConvergence}
V_0(C(T,K)|\pi,n)&\longrightarrow V_0(C(T,K)|\sigma_0)\\
V_0(P(T,K)|\pi,n)&\longrightarrow V_0(P(T,K)|\sigma_0)
\end{split}
\mbox{ , as }n\to\infty,
\end{align}
where the convergence takes place $P_{\sigma_0}$-almost surely. 
\end{theorem} 
\begin{proof}
Note that we can write the density $f_n(\sigma^2)$ in the form 
\begin{equation}
f_n(\sigma^2)=\frac{e^{-n h_n(\sigma^2)}\pi(\sigma^2)}{\int_{\R_+}e^{-n  h_n(\sigma^2)}\pi(\sigma^2) d\sigma^2} \mbox{ with } h_n(\sigma^2)=\frac{1}{2}\left(\frac{\hat\sigma^2_n}{\sigma^2}+\log(\sigma^2)\right).
\end{equation} 
We set $h(\sigma^2)=\frac{1}{2}\left(\frac{\hat\sigma^2_0}{\sigma^2}+\log(\sigma^2)\right)$ which has a unique minimum in $\sigma^2_0$ with value $\beta_0=\frac{1}{2}(1+\log(\sigma_0^2))$. We therefore identify \eqref{eqa:BSpricesConvergence} as a saddle point problem in the sense of Appendix \ref{app:A} with $g(\sigma^2)$ given by the expression in the brackets $[\ldots]$ in \eqref{eqa:BSpricesSubjective}. As remarked earlier, this $g(\sigma^2)$ fulfils the conditions of Lemma \ref{lem:BayesConsistBoundedPrior}.

As we wish  to apply Lemma \ref{lem:BayesConsistBoundedPrior} with $\Theta=\R_+$ and $\theta=\sigma^2$, we have to verify the remaining conditions. 
By the law of large numbers, $\hat\sigma_n^2\to\sigma_0^2$ $P_{\sigma_0}$ almost surely, it is easily seen that $h_n(\theta)\to h(\theta)$ uniformly of compact sets in $\R_+$ holds $P_{\sigma_0}$- almost surely.

Next we choose the function $a(\sigma^2)=1_{\{\sigma^2>1\}}2\log(\sigma^2)$.  $e^{-a(\theta)}$ is bounded by $1$ and decays like $\frac{1}{\sigma^4}$ for large $\sigma^2$ and thus is integrable with respect to $d\sigma^ 2$. Let us consider the functions $\tilde h_n(\sigma^2)$ for $n\geq 5$
\begin{equation}
\label{eqa:EstH}
\tilde h_n(\sigma^2)=\frac{1}{2}\left(\frac{\hat\sigma^2_n}{\sigma^2}+\log(\sigma^2)-\frac{1_{\{\sigma^2>1\}}4\log(\sigma^2)}{n}\right)\geq \frac{1}{2}\left(\frac{\underline{\sigma}^2}{\sigma^2}+\left(1-1_{\{\sigma^2>1\}}\frac{4}{5}\right)\log(\sigma^2)\right),
\end{equation}
with $\sigma^2_0\geq\underline{\sigma}^2=\inf_{n\geq 5}\hat \sigma^2_n>0$ $P_\sigma$-a.s.. The positivity of $\underline{\sigma}^2$ follow from $\hat\sigma_n^2>0$  and $\hat \sigma_n^2\to\sigma_0^2>0$ ($P_{\sigma_0}$ a.s.) We now set $n_0=5$, $\gamma=1$ and we construct the environment $U(\sigma_0^2)=(l_-,l_+)$ such that $l_-<\min(1,\sigma_0^2)$ is sufficiently small such that $\frac{1}{2}\left(  \frac{\underline{\sigma}_5^2}{\sigma^2}+\log(\sigma^2)\right)>\beta_0+1$ and $l_+>\max(1,\sigma_0^2)$ sufficiently large such that $\frac{1}{5}\log(l_+)\geq \beta_0+1$. It is an easy consequence of \eqref{eqa:EstH} that $\tilde h_n(\sigma^ 2)\geq \beta_0+1$ for $\sigma^ 2 \in \R_+\setminus U(\theta_0)$.  As this was the last condition from Lemma \ref{lem:BayesConsistBoundedPrior}, the statement in \eqref{eqa:BSpricesConvergence} follows. 
\end{proof}

\subsection{A numerical example  for the BS market}
We provide a numerical example for the dynamics of the subjective BS price with non informative prior $\pi(\sigma^2)=1$ according to \eqref{eqa:BSpricesSubjective}. The initial value of the fictitious asset is fixed to $S_0=100$, the annual volatility is set to $\sigma_0=15.8\%$, and the drift is $\rho=0.002$ per year. Strike prices $K$ for the European call with 3 month maturity are computed from $K=80$ to $K=134$ as a function of the number of observations $N$ ranging from 2 to $30$ or $150$, respectively. Observations of a realization of the BS market are simulated.

The simulation is carried out using \texttt{R 3.3.1}, the integrals in \eqref{eqa:BSpricesSubjective} are carried out using a 1-d adaptive numerical quadrature implemented in the \texttt{R} function \texttt{integrate}.

\begin{figure}[t]
\centerline{\includegraphics[scale=.3]{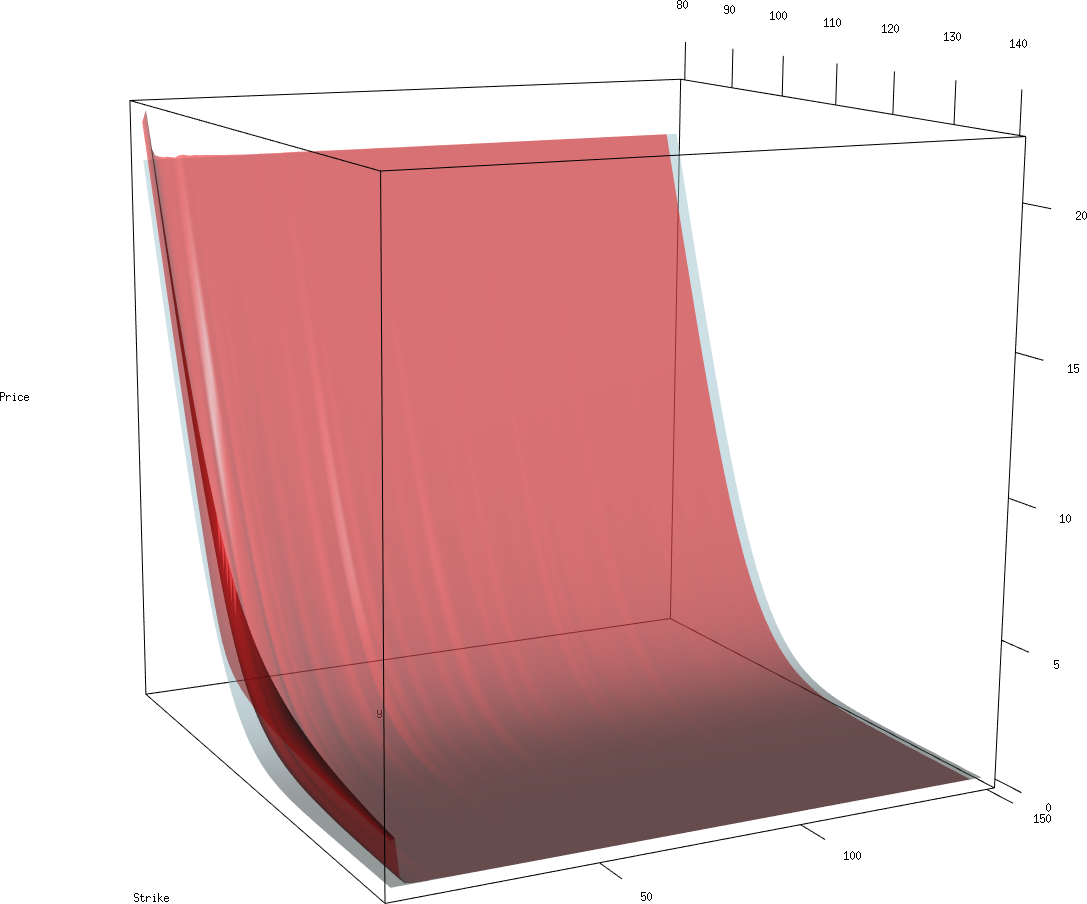}}
\caption{\label{fig:BS} Pricing surface (z-axis, to the top) for the subjective BS price for an European call with 3 month maturity as a function of the strike $K$ (y-axis, to the front) and Number of observations (x-axis, to the right). The lightgray surface gives the BS reference price for the same strike price.}
\end{figure}

Figure \ref{fig:BS} displays the price surface of the European call for fixed maturity time and varying strike and number of observations. The deviation in price from the standard BS price is already quite low after 30 observations, which is a usual number of observations  for a short term close-to-close volatility estimator, confer Figure \ref{fig:BSCONV}.   For 150 observations, which is close to the number of observations commonly used in a a long term close-to-close volatility estimator, the difference in price is less than 5\% in the given scenario, however for a short term 20 day volatility estimator it is around 25\%. This confirms the relevance of the convergence analysis in Theorem \ref{theo:BayesConsistency} with respect to high frequency observations. But  the usual day-to-day estimations  of historic volatility are not sufficiently 'high frequency' to neglect the difference in price caused by the measurement error of this quantity. This of course confirms previous studies on Bayesian option pricing in a time series context. Note however that the number of observations can, at least theoretically, be increased arbitrarily during a single day using intra day data.

\begin{figure}[t]
\centerline{\includegraphics[scale=.6]{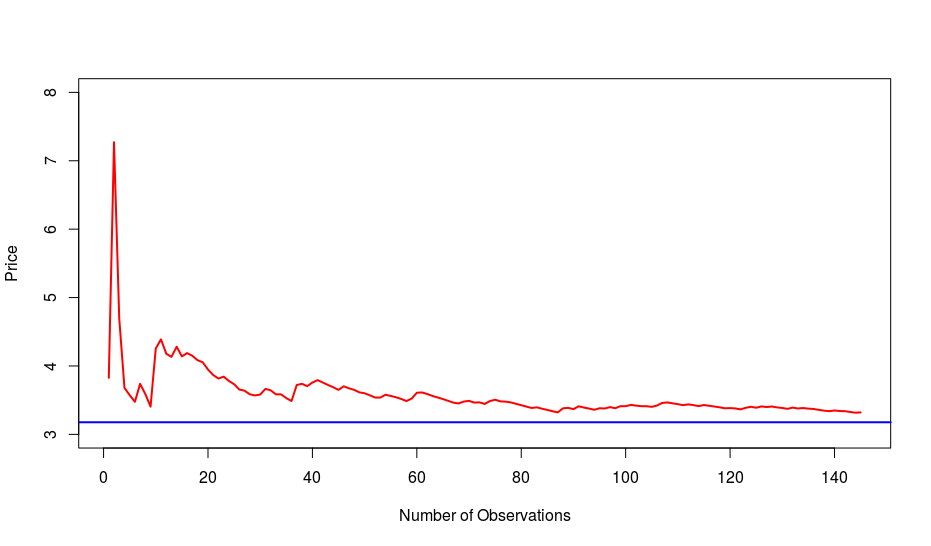}}
\caption{\label{fig:BSCONV} Convergence of the subjective BS price (red) to the actual BS price (blue) as a function of the number of observations. The strike is $K=S_0=100$ EUR.}
\end{figure}

\section{The Merton Market with Unknown Jump Distribution}

\subsection{Some fundamentals on the Merton market}
In this section we consider the Merton market \cite{Mer2} as a simple example of a market of exponential L\'evy type. We thus consider a market dynamics where a jump term of compound Poisson type is added:
\begin{equation}
\label{eqa:Merton}
S_t=S_0e^{X_t} \mbox{ with }X_t=\rho t+\sigma W_t+\sum_{j=\pm 1}^ {N_t}Y_j.
\end{equation}
Here $N_t$ is a Poisson process with intensity $\lambda$ and $Y_j$ are i.i.d.\ random variables. We assume $N_0=0$ and that $N_t$ takes non positive integer values on $[-\tau,0]$, in which case the summation starts with $-1$ and goes downward.  All components of \eqref{eqa:Merton} are independent from each other.

 Let $\nu_1$ denote the probability  measure on $(\R,{\cal B}(\R))$ such that $Y_j\sim \nu_1$. We assume $\nu_1(\{0\})=0$ and that $\mathscr{L}_{\nu_1}(\alpha)=\int_\R e^{y\alpha}d\nu_1(y)<\infty$ for all $\alpha\in[0,1]$. The L\'evy measure associated with \eqref{eqa:Merton} is $\nu=\lambda\nu_1$. To obtain  the Merton model, we set $\nu_1=N(m,\delta^2)$ where $m\in\R$ and $\delta^2>0$. Let $P_\theta$, $\theta=(\lambda,\delta^2,m)$.  be a measure on $(\Omega,(\F^+_t)_{t\in[-\tau,T]})$ such that $S_t$ and $X_t$ are adapted processes with the given distribution. Here we omitted the parameters $\rho,\sigma$ from $\theta$, as they are either given by public data or can (in the idealized world set by the model) be determined without estimation error by continuous ('high frequency') observations, respectively.

Let us next proceed to option pricing in the Merton model. Here two options are frequently chosen, mean correction and Esscher transformation \cite{CT,Sato} and \cite{Iac} . However, there are infinitely many further options to construct an equivalent martingale measure $Q_\theta$. Here we choose mean correction, for simplicity. We set
\begin{equation}
\label{eqa:DriftCorrectionMerton}
\mu_\theta=-\frac{1}{2}\sigma^ 2-\lambda \left(e^ {\frac{1}{2}\delta^ 2+m}-1\right)
\end{equation}
and we obtain a martingale measure applying Grisanov's formula to the Gaussian part of the market measure $P_\theta$
\begin{equation}
\label{eqa:MertonMartingaleMeasure}
Q_\theta=L_TP_\theta \mbox{ with }L_T=e^ {-\frac{\mu_\theta}{\sigma} W_t-\frac{1}{2}\left(\frac{\mu_\theta}{\sigma}\right)^ 2}.
\end{equation}
Using $Q_\theta$ for option pricing, we obtain the following expression for the European call and put \cite[10.1 Merton's approach]{CT}
\begin{align}
\begin{split}
\label{eqa:MertonMCprices}
V_0^{\mathrm{MC}}(C(T,K)|\theta)&=e^{-\lambda T}\sum_{n=0}^\infty\frac{(\lambda T)^n}{n!} V_0\left(C(K-nm,T)|\sqrt{\sigma^2+\frac{n}{T}\delta^2}\right)\\
V_0^ {\mathrm{MC}}(P(T,K)|\theta)&=e^{-\lambda T}\sum_{n=0}^\infty\frac{(\lambda T)^n}{n!} V_0\left(P(K-nm,T)|\sqrt{\sigma^2+\frac{n}{T}\delta^2}\right)
\end{split} 
\end{align}
The mean correction (MC) Merton prices are again non negative and bounded by $S_0$ and $K$, respectively. Furthermore, for $\theta\in\Theta=\R_+\times\R_+\times \R$ they continuously depend on the parameter $\theta=(\lambda,\delta^ 2,m)$. This is easily seen using the continuity and uniform boundedness of each expression $V_0(\ldots)$ on the right hand side of \eqref{eqa:MertonMCprices} and applying the theorem of dominated convergence  to the sum. Note that local bounds for $\lambda$ can be used to construct a dominating sequence.

\subsection{Pricing with an unknown jump intensity and jump distribution}
We now consider a market participant who believes that the Merton model \eqref{eqa:Merton} is structurally correct. Furthermore, she or he observed the asset quota $S_t$ for $t\in[-\tau,0]$ continuously. We assume that the path $(S_t)_{t\in[-\tau,0]}$ is generated by the model \eqref{eqa:Merton} with parameters $\rho,\sigma_0$ and $\theta_0=(\lambda_0, \delta_0^ 2,m_0)$, where $\theta_0$ is unknown. 

Let $\nu_{\theta}$ be the Levy measure associated with the parameters $\theta\in\Theta$ and let $\nu_{\theta,1}$  be the normalized L\'evy measure. We recall the Grisanov formula for compound Poisson processes which can e.g. be found in \cite[Chapter 5.4.3]{App} or \cite[Chapter 10]{CT}. We define define the measure $P_\theta$ on the sigma algebra $\F_0$ containing the information in the time interval $[-\tau,0]$
\begin{equation}
\label{eqa:GrisanovCP}
P_\theta = L_\tau(\theta|\theta_0)P_{\theta_0} \mbox{ with } L(\theta|\theta_0)=\exp \left\{\tau(\lambda_0-\lambda)+\sum_{-\tau\leq s\leq 0}\log\left(\frac{d\nu_\theta}{d\nu_{\theta_0}}(\Delta X_s)\right)\right\},
\end{equation}
where we use the convention  $\log\left(\frac{d\nu_\theta}{d\nu_{\theta_0}}(0)\right)=0$ and $\Delta X_s=X_s-X_{s-}$ is the jump height observed at time $s$. It is then well known, that $(X_t)_{t\in[-\tau,0]}$ follows the dynamic \eqref{eqa:Merton} with $\theta=(\lambda,\delta^ 2,m)$. 

As one usually does in the statistics of continuous processes, we interpret $L(\theta|\theta_0)$ as the likelihood of $\theta$ with respect to some fixed background measure $P_{\theta_0}$. As the dependency of $\theta_0$ drops out in maximum likelihood estimates and in the Bayesian formalism, as long as $P_\theta$ is absolutely continuous with respect to $P_{\theta_0}$, we can without loss of generality choose the true parameter set $\theta_0$ for the reference measure, even though $\theta_0$ is not known.

Let $\pi(\theta)$ be some continuous, bounded prior on $\Theta$.   The a posteriori density is then defined as
\begin{equation}
\label{eqa:PostLevy}
f_\tau(\theta)=f_\tau(\theta|\pi)=\frac{L(\theta|\theta_0)\pi(\theta)}{\int_\Theta L(\xi|\theta_0)\pi(\xi)\, d\xi}.
\end{equation}
Here again the dependency on the observed path $(S_t)_{t\in[-\tau,0]}$ is suppressed. The following lemma gives a more explicit formula for the a posteriori distribution:
\begin{lemma}[A Posteriori Distribution for the Merton Model] 
\label{lem:PostMerton}
Let $Y_1=\Delta X_{s_1},\linebreak\ldots, Y_{N}=\Delta X_{s_1}$ where $-\tau\leq s_1<s_2<\ldots< s_N\leq 0$ are the observed jump heights and $N=-N_{-\tau}$ is the observed number of jumps from $t=-\tau$ up to time $t=0$. Let $\hat\delta_N^2=\frac{1}{n}\sum_{j=1}^N(Y_j-\hat m_N)^2$ with $\hat m_N=\frac{1}{N}\sum_{j=1}^NY_j$. Furthermore let $\hat \lambda=\frac{N}{\tau}$. Then, if $N\geq 2$,
\begin{equation}
\label{eqa:PostMerton}
f_\tau(\theta)=\frac{e^{-N\frac{\lambda}{\hat \lambda_N}}\lambda^N \frac{1}{\delta^N}e^{-\frac{N}{2}\left(\frac{\hat\delta_N^2}{\delta^2}+\left(\frac{\hat m_N-m}{\delta}\right)^2\right)}\pi(\lambda,\delta^2,m)}{\int_{\R_+^2\times \R} e^{-N\frac{\bar\lambda}{\hat \lambda_N}}\bar{\lambda}^N \frac{1}{\bar{\delta}^N}e^{-\frac{N}{2}\left(\frac{\hat\delta_N^2}{\bar{\delta}^2}+\left(\frac{\hat m_N-\bar m}{\bar \delta}\right)^2\right)}\pi(\bar \lambda,\bar{\delta}^2\bar ,m)\,d\bar \lambda d\bar{\delta}^2 d\bar m}.
\end{equation}
\end{lemma}
\begin{proof}
Note that for the Merton model and $Y\not=0$, we have
\begin{align}
\label{eqa:RadonNikodymLevyMeasures}
\begin{split}
\log\left(\frac{d\nu_\theta}{d\nu_{\theta_0}}(Y)\right)&=\log(\lambda)-\log(\delta)-\frac{1}{2}\left(\frac{Y-m}{\delta}\right)\\
&-\log(\lambda_0)+\log(\delta_0) +\frac{1}{2}\left(\frac{Y-m_0}{\delta_0}\right).
\end{split}
\end{align}
Inserting this into the change of measure formula \eqref{eqa:GrisanovCP} with $Y_j$ in the place of $Y$, we note that exactly $N$ such terms occur in the exponent. Now \eqref{eqa:PostMerton} follows by a straight forward reordering of terms and the observations that terms depending on $\theta_0$ drop out in \eqref{eqa:PostMerton} due to normalization. 
\end{proof}

The following Definition and theorem now follow the same lines as in the BS case. Note however that despite the assumption of a continuous observation in the time interval $[-\tau,0]$, the subjective market measure and the subjective pricing measures in this case differ from the standard Merton market and pricing measures.

\begin{definition}[Subjective Merton Market and Pricing Measure]
\label{def:subjMertonMeasures}
Let $P_\theta$ and $Q_\theta$ be the measures that define the Merton market and the mean corrected Merton pricing measures, respectively. Then, given a bounded, continuous prior $\pi(\theta)$ and the continuous observations $(S_t)_{t\in[-\tau,0]}$ of the past, the subjective Merton market measure $P_\pi$ and the subjective Merton mean correction pricing measure are defined as
\begin{equation}
P_\pi(A)=\int_\Theta P_\theta(A)f_\tau(\theta|\pi)\, d\theta \mbox{ and }Q_\pi(A)=\int_\Theta Q_\theta(A)f_\tau(\theta|\pi)\, d\theta,~~A\in\F^+_T.
\end{equation} 
\end{definition} 
We note that the kernels $P_\theta(A)$ are measurable in $\theta$: in fact, due to \eqref{eqa:GrisanovCP} and Lebesgue's theorem of dominated convergence, these expressions are ven continuous in $\theta$.
\begin{theorem}
\label{theo:NoArbitrageMerton}
Let $\pi(\theta)$ and $\pi'(\theta)$, $\theta\in\Theta$ be two prior functions such that  $\pi(\theta)d\theta$ and $\pi'(\theta)d\theta$ are equivalent measures. Then the subjective mean corrected Merton pricing measure $Q_{\pi'}$ is an equivalent martingale measure to the subjective Merton market measure $P_\pi$.  
\end{theorem}
\begin{proof}
The proof is completely analogous to the proof of Theorem \ref{theo:NoArbitrageBS}.
\end{proof}

\begin{corollary}[Subjective Merton Option Prices]
Arbitrage free MC prices for the subjective Merton market model $P_\pi$ are
\begin{align}
\begin{split}
\label{eqa:SubjMertonMCprices}
V_0^{\mathrm{MC}}(C(T,K)|\pi,\tau)&=\int_\Theta V_0^{\mathrm{MC}}(C(T,K)|\theta)\, f_\tau(\theta|\pi)\,d\theta\\
V_0^{\mathrm{MC}}(P(T,K)|\pi,\tau)&=\int_\Theta V_0^{\mathrm{MC}}(P(T,K)|\theta)\, f_\tau(\theta|\pi)\,d\theta
\end{split} 
\end{align}
where $ V_0^{\mathrm{MC}}(\cdots|\theta)$ are given in \eqref{eqa:MertonMCprices} and $f_\tau(\theta|\pi)$ is given in \eqref{eqa:PostMerton} for $\theta=(\lambda,\delta^2,m)$. As usual, the dependence on $(S_t)_{t\in[\tau,0]}$ is suppressed in the notation.
\end{corollary}
\subsection{The limit of long observation time}
We have seen that also in the case of high frequency observations, a considerable insecurity on the proper calibration of the Merton model prevails. In this section we consider the limit, when the market dynamics is unchanged since a time $\tau$, we posses continuous observations since that time, and we consider the limit of long observation times $\tau\to\infty$. As volatility levels change fundamentally on a time scale of a few years and only a hand full major jump events are observed during a year, it is questionable if the $\tau\to\infty$ limit is of practical importance. Nevertheless we prove the following Bayesian consistency result for the European options priced with the subjective MC Merton price formula \eqref{eqa:SubjMertonMCprices}.

\begin{theorem}[Convergence of Subjective Merton MC Option Prices]
Let $\theta_0\in\Theta$, $\theta_0=(\lambda_0,\delta_0^2,m_0)$, be the  set of parameters such that $S_t$ follows the dynamics \eqref{eqa:Merton}. $\theta_0$ however is unknown to a market participant, who prices Europen options according to \eqref{eqa:SubjMertonMCprices} with some continuous, bounded prior such that $\pi(\theta_0)>0$. 

In the limit of large observation time, $\tau\to\infty$, the subjective MC Merton prices for the European call and put options converge $P_{\theta_0}$ almost surely to the MC Merton prices with parameter set $\theta_0$, i.e.
\begin{align}
\begin{split}
\label{eqa:MertonPricesConvergence}
V_0^{\mathrm{MC}}(C(T,K)|\pi,\tau)&\longrightarrow V_0^{\mathrm{MC}}(C(T,K)|\theta_0)\\
V_0^{\mathrm{MC}}(P(T,K)|\pi,\tau)&\longrightarrow V_0^{\mathrm{MC}}(P(T,K)|\theta_0)
\end{split}
\mbox{ , as }\tau\to\infty.
\end{align}
\end{theorem} 
\begin{proof}
We write the problem in saddle point form, see Appendix \ref{app:A}. First, the functions $g(\theta)$ are given in \eqref{eqa:MertonMCprices}. As discussed, these functions are bounded and continuous. 

the function $f_\tau(\theta)$ can be rewritten as $e^{-Nh_N(\theta)}\pi(\theta)/\int_\Theta e^{-Nh_N(\bar \theta)}\pi(\bar \theta)d\theta$ with
\begin{equation}
h_N(\theta)=\left(\frac{\lambda}{\hat\lambda_N}-\log(\lambda)\right)+\frac{1}{2}\left(\frac{\hat\delta_N^2}{\delta^2}+\log(\delta^2)+\left(\frac{\hat m_N-m}{\delta}\right)^2\right),~~\theta=(\lambda,\delta^2,m)\in\Theta.
\end{equation}  
Replacing the estimated quantities $\hat\lambda_N$, $\hat \delta^2_N$ and $\hat m_N$ with $\lambda_0$, $\delta_0^2$ and $m_0$, we obtain the function $h(\theta)$. It is easily verified that $h(\theta)$ is minimal at $\theta=\theta_0$, where it attains the value $\beta_0=(1-\log(\lambda_0))+\frac{1}{2}(1+\log(\delta_0)^2)$. 

We have $N=-N_{-\tau}\to\infty$ as $\tau\to\infty$ $P_{\theta_0}$-a.s. and therefore $\hat\lambda_N\to\lambda_0$, $\hat\delta^2_N\to\delta^2_0$ and $\hat m_N\to m_0$ $P_{\theta_0}$-almost surely by the strong law of large numbers. It is thus easily checked that $h_N(\theta)\to h(\theta)$ uniformly on compact sets holds almost surely.

We next define the auxiliary function $a(\theta)$ from the assumptions of Lemma \ref{lem:BayesConsistBoundedPrior}. A possible choice is
\begin{equation}
\label{eqa:AuxiliaryA}
a(\theta)=1_{\{\lambda\geq 1\}}2\log(\lambda)+1_{\{\delta^ 2>1\}}2\log(\delta^ 2)+1_{\{|m|>1\}}2\log(m).
\end{equation}
Let $\bar \lambda=\sup_{\tau\geq \tau^*} \hat\lambda_N(\tau)<\infty$ with a stopping time $\tau^*$ sufficiently large that in $[-\tau^*,0]$ there occurs at least fife jumps.  $\tau^*< \infty$  holds $P_{\theta_0}$ almost surely. Furthermore set $  \underline{\delta}^2=\inf_{\tau\geq \tau^*}\hat \delta_N^2>0$ and $  \bar{\delta}^2=\sup_{\tau\geq \tau^*}\hat \delta_N^2<\infty$, $\underline{m}=\inf_{\tau\geq \tau^*}\hat m_N\in\R$ and finally $\bar{m}=\sup_{\tau\geq \tau^*}\hat m_N\in\R$ . All these statements have to be understood in the $P_{\theta_0}$ a.s. sense. We see with a similar argument as in the proof of Theorem \ref{theo:BayesConsistency} that the following estimate is uniform in $\tau\geq \tau_5$:
\begin{align}
\label{eqa:LowerEstimate}
\begin{split}
\tilde h_N(\theta)=h_N(\theta)-\frac{a(\theta)}{N}&\geq \left(\frac{\lambda}{\bar \lambda}-\left(1+\frac{21_{\{\lambda\geq 1\}}}{5}\right)\log(\lambda)\right)\\
&+\frac{1}{2}\left(\frac{\underline{\delta}^2}{\delta^2}+
\left(1-\frac{41_{\{\lambda\geq 1\}}}{5}\right)\log(\delta^2)\right)\\
&+\frac{1}{2}\left(\left(\frac{\mathrm{dist}(m,[\underline{m},\bar m])}{\bar \delta}\right)^2-\frac{4}{5}1_{\{|m|>1\}}\log(|m|)\right)
\end{split}
\end{align}
We chose $\gamma=1$. We now construct the bounded open environment $U(\theta_0)$ from the above  estimate. First, chose $l_\pm^\delta$ as in the proof of Theorem \ref{theo:BayesConsistency}, however with $\beta_0=h(\theta)$. Then the middle term exceeds $\beta_0+1$ if $\delta^2\in\R_+\setminus[l_-^\delta,l_+^\delta]$. 

Secondly, chose $0<l_-^\lambda<\min\{1,\lambda_0\}$, then the first therm on the right hand side of \eqref{eqa:LowerEstimate} is positive for $0<\lambda <l_-^\lambda$ as $\log(\lambda )<0$ in this case. If we chose $l_+^\lambda>\max\{1,\lambda_0\}$ sufficiently large such that $\lambda>\bar\lambda \frac{7}{5}\log(\lambda)$ holds for $\lambda>l_+^\lambda$, the first term on the right hand side is larger zero also for such $\lambda$. Thus it is bounded by zero for $\lambda\in \R_+\setminus [l_-^\lambda,l_+^\lambda]$. 

Finally choose $l^m>|m_0|$ sufficiently large such that $ \left(\frac{\mathrm{dist}(m,[\underline{m},\bar m])}{\bar \delta}\right)^2\geq \frac{4}{5} \log(|m|)$ if $|m|>l^m$. If $m\in\R\setminus [-l^m,l^m]$, then the third term in \eqref{eqa:LowerEstimate} is positive as well. Thus we can choose the environment $U(\theta_0)=(l_-^\lambda,l_+^\lambda)\times (l_-^\delta,l_+^\delta)\times (-l^m,l^m)$ and obtain that $\tilde h(\theta)\geq \beta_0+1$ for $\theta\in\Theta\setminus U(\theta_0$. As this was the last assumption of Lemma \ref{lem:BayesConsistBoundedPrior} to verify, the convergence statement \eqref{eqa:MertonPricesConvergence} now follows from Lemma \ref{lem:BayesConsistBoundedPrior}.  
\end{proof}

 \subsection{A numerical example for the Merton market}

\begin{figure}[t]
\centerline{\includegraphics[scale=.5]{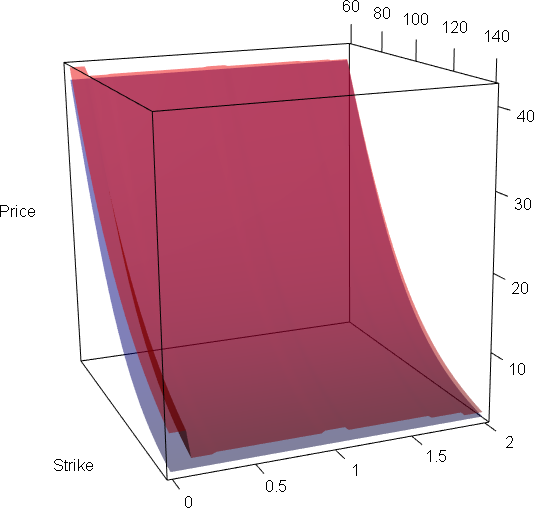}}
\caption{\label{fig:MertonSurface} Pricing surface in the Merton Model (grey) and its Bayesian extension (red) as a function of the observation time (in years). Discontinuities in time of the price surface occur at times when jumps occur in the underlying asset price and new information on the jump distribution is revealed.}
\end{figure}

Again we use a noninformative prior for the Merton Market. We leave the data for the BS-part of the market as in subsection 2.4, but we add jumps with a Gaussian distribution of jumps $N(m,\delta^2)$ with $m=0$ and $\delta=0-05$. The jump frequency ist set to $\lambda= 4$, which corresponds to four jump events per year in average.   Option prices are calculated with Merton's mean correction. Figure \ref{fig:MertonSurface} shows the results for a strike range $K\in [60,140]$ which is based on a simulated trajectory of the underlying Merton model.

\begin{figure}[t]
\centerline{\includegraphics[scale=.5]{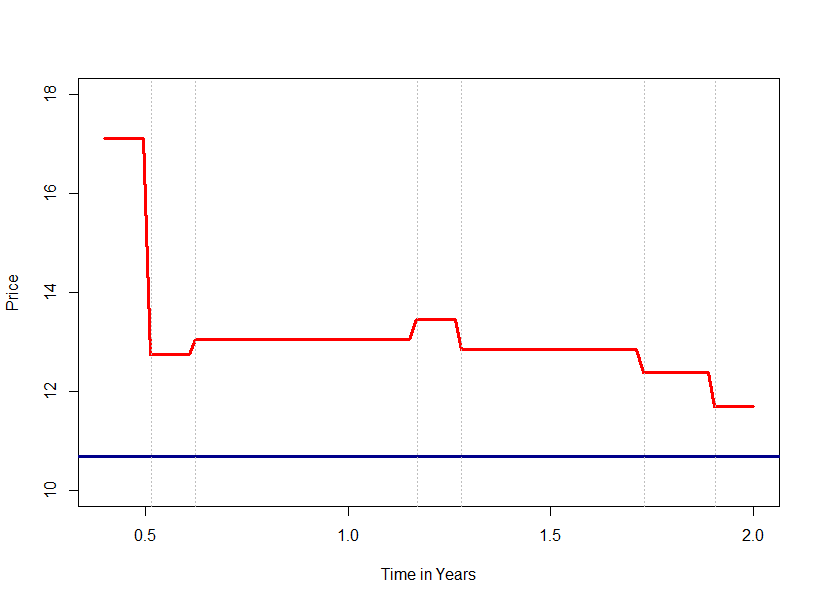}}
\caption{\label{fig:MertonConvergence} Convergence of the Bayesian Option Price in the Merton Model over observation time over 2 years for the strike $K=100$. The six vertical lines correspond to jump events of the underlying. The final difference between the Merton price remains in the range of 1 Euro or 10\% of the option price.  }
\end{figure}

Figure \ref{fig:MertonConvergence} provides a section through the pricing surface for a strike $K=100$. The found difference in the option priced between a Merton and a Bayesian Merton option price at the end of the  observation period of two years still amounts to approximately 10\% of the option's price, which of course, is a significant deviation.  

\section{Conclusion}
In the present work, we have given a systematic and mathematically rigorous account on Bayesian methods in option pricing for the case of the Black-Scholes and the Merton market. In particular, we proved the existence of equivalent martingale measures and market incompleteness for these Bayesian market models. 

 Bayesian corrections to option prices due to uncertainty in volatility estimates has been intensly studied in the context of time series models, see e.g. \cite{DS,FGMM,FMW,GHM,HLM,JP,RHBF,RS}. We have shown that this concepts  crucially depends on the observation frequency. In particular, as a consequence of Bayesian consistency, it is obsolete in the context of high frequency observations. 

In contrast, Bayesian prices differ from non Bayesian ones even in the case of continuous observations, if Market models of exponential L\'evy type are considered. The reason for this crucial difference to the BS-market lies in the fact that one can not increase the information on the jump distribution by simply increasing the observational frequency. The uncertainty on the jump distribution thus prevails over a span of time of several years and has a significant impact to option prices. Despite also this difference converges to zero in the observation time is sent to infinity, this mathematical result is not very relevant as the statistical properties of asset markets are certainly non stationary on a time scale of several years. 

Despite Bayesian methods in asset pricing have been predominantly applied to volatility estimation in markets with Gaussian log-returns, we here suggest that Bayesian estimation and option pricing in exponential L\'evy markets is an even more interesting application of the Bayesian approach in finance.

\appendix

\section{The Saddle Point Method for Bayesian Consistency}
\label{app:A}
Here we give a variant of the saddle point method that is tailored for the  Bayesian consistency for option prices, see \cite{CR,Gho} for comprehensive reviews on Bayesian consistency.

Let $\Theta\subseteq \R^d$ be some open region and let $\pi(\theta)$ be some continuous, bounded and non negative function on $\Theta$ such that $\int_\Theta\pi(\theta)d\theta<\infty$. Let $h,h_n:\Theta\to\R$ be continuous functions that are bounded from below such that $h_n(\theta)\to h(\theta)$ uniformly on compact sets. Furthermore, $h(\theta)$ assumes a unique global minimum, $\theta_0\in\Theta$ where it attains the value $\beta_0=h(\theta_0)$. Also, there exists a positive number $\gamma>0$ and an open, bounded neighbourhood $U(\theta_0)$, $\bar U(\theta_0)\subseteq \Theta$, and  $n_0\in\N$ such that for all $\theta\in \Theta\setminus U(\theta_0)$ we have $h_n(\theta)>\beta_0+\gamma$ $\forall n\geq n_0$.  Finally we assume that $\pi(\theta_0)>0$.  

\begin{lemma}[Saddle Point Method with Integrable Prior]
\label{lem:SaddlePointIntegrable}
Let $h,h_n$ be as described above and let $g:\Theta\to\R$ be a bounded and continuous function. Then,
\begin{equation}
\label{eqa:SaddlePoint}
\frac{\int_\Theta e^{-n h_n(\theta)} g(\theta)\pi(\theta)\,d\theta}{\int_\Theta e^{-n h_n(\theta)} \pi(\theta)\,d\theta}\longrightarrow g(\theta_0),\mbox{ if }n\to\infty.
\end{equation}
\end{lemma}
\begin{proof}
We start the proof with the following bound on the convergence speed of the nominator to zero: Let $\beta > \beta_0$, and let  $\bar{\gamma}=\frac{\beta-\beta_0}{2}>0$. As $h_n\to h_n$ uniformly on compact sets, there exists a $\varepsilon >0$ and a number $n_1\in\N$ such that $|h_n(\theta)-h(\theta)|<\bar{\gamma}/2$ for all $\theta\in \bar{B}_\varepsilon(\theta_0)$ and $n\geq n_1$. Here $B_\varepsilon(\theta_0)$ is the ball around $\theta_0$ with radius $\varepsilon$. Thus, $h_n(\theta)-\beta < - \frac{\bar \gamma}{2}$ on $\bar{B}_\varepsilon(\theta_0)$. We now get
\begin{equation}
\label{eqa:BoundNomEstim}
e^{n\beta}\int_\Theta e^{-nh_n(\theta)}\pi(\theta) \, d\theta\geq e^{n\frac{\bar \gamma}{2}}\int_{ B_\varepsilon(\theta_0)}\pi(\theta)d\theta\longrightarrow \infty \mbox{ as }n\to\infty.
\end{equation}
Here we also used that $\pi(\theta_0)>0$ and thus the integral of $\pi(\theta)$ over  $ B_\varepsilon(\theta_0)$ is positive.

Let next be $\varepsilon >0$ arbitrary. Let $\bar \gamma=\inf_{\theta\in \Theta\setminus  B_\varepsilon(\theta)}[h(\theta)-\beta_0]>0$ as the minimum is unique and $h(\theta)>\beta_0+\gamma$ on $\theta\in \Theta\setminus U(\theta_0)$ with $\gamma$ and $U(\theta_0)$ as in the assumptions. If $n>\max\{n_0,n_1\}$, with $n_0$ from the assumptions and $n_1$ sufficiently large such that $|h_n(\theta)-h(\theta)|\leq \min(\gamma,\bar\gamma)/2$ on the compact set $\bar U(\theta_0)\setminus B_\varepsilon(\theta_0)$, we see that for such $n$ $h_n(\theta)>\beta=\beta_0+ \min(\gamma,\bar\gamma)/2$ on $\Theta\setminus  B_\varepsilon(\theta_0)$. Consequently,
\begin{equation}
\label{eqa:EstNominator}
\int_{\Theta\setminus B_\varepsilon(\theta_0)} e^{-nh_n(\theta)}g(\theta)\pi(\theta) \, d\theta \leq  \left(\sup_{\theta\in\Theta}|g(\theta)|\int_\Theta\pi(\theta)\,d\theta\right)e^{-\beta n}.
\end{equation}
Combining this with \eqref{eqa:BoundNomEstim}, we obtain that
\begin{equation}
\label{eqa:EstFrac}
\frac{\int_{\Theta\setminus B_\varepsilon(\theta_0)} e^{-nh_n(\theta)}g(\theta)\pi(\theta) \, d\theta}{\int_\Theta e^{-nh_n(\theta)}\pi(\theta) \, d\theta}\longrightarrow 0 \mbox{ as }n\to\infty.
\end{equation}
Therefore, applying \eqref{eqa:EstFrac} once for $g(\theta)$ itself and one for $g(\theta)$ replaced with one, we get from the fact that adding sequences converging to zero do not change the $\limsup$ or the $\liminf$ 
\begin{align}
\label{eqa:EstLimInf}
\begin{split}
&\liminf_n\frac{\int_{\Theta} e^{-nh_n(\theta)}g(\theta)\pi(\theta) \, d\theta}{\int_\Theta e^{-nh_n(\theta)}\pi(\theta) \, d\theta}\\
&=\liminf_n\frac{\int_{B_\varepsilon(\theta_0)} e^{-nh_n(\theta)}g(\theta)\pi(\theta) \, d\theta}{\int_\Theta e^{-nh_n(\theta)}\pi(\theta) \, d\theta}\\
&\geq \inf_{\theta\in B_\varepsilon(\theta_0)}g(\theta) \liminf_n\frac{\int_{B_\varepsilon(\theta_0)} e^{-nh_n(\theta)}\pi(\theta) \, d\theta}{\int_\Theta e^{-nh_n(\theta)}\pi(\theta) \, d\theta}\\
&=\inf_{\theta\in B_\varepsilon(\theta_0)}g(\theta)\liminf_n\frac{\int_{\Theta} e^{-nh_n(\theta)}\pi(\theta) \, d\theta}{\int_\Theta e^{-nh_n(\theta)}\pi(\theta) \, d\theta}=\inf_{\theta\in B_\varepsilon(\theta_0)}g(\theta).
\end{split}
\end{align}
Likewise, we prove 
\begin{equation}
\label{eqa:EstLimSup}
\limsup_n\frac{\int_{\Theta} e^{-nh_n(\theta)}g(\theta)\pi(\theta) \, d\theta}{\int_\Theta e^{-nh_n(\theta)}\pi(\theta) \, d\theta}\leq \sup_{\theta\in B_\varepsilon(\theta_0)}g(\theta).
\end{equation} 
As these inequalities \eqref{eqa:EstLimInf} and \eqref{eqa:EstLimSup} hold for arbitrary $\varepsilon>0$, we can take the supremum over $\varepsilon>0$ in \eqref{eqa:EstLimInf} and the infimum over $\varepsilon>0$ in \eqref{eqa:EstLimSup}. By continuity of $g(\theta)$, we obtain
$g(\theta_0)$ as upper bound for the $\limsup$ in \eqref{eqa:EstLimSup} and as lower bound for the $\liminf$ in \eqref{eqa:EstLimInf} from which the convergence \eqref{eqa:SaddlePoint}  follows.
\end{proof}

The following Lemma deals with some modification of the previous, in order to deal with the case of a non integrable prior, like e.g. the non informative prior:

\begin{lemma}[Saddle Point Method with Non Integrable Prior]
\label{lem:BayesConsistBoundedPrior}
Consider the situation in the beginning of the appendix, where however the prior $\pi(\theta)$ is  continuous and bounded, but not necessarily integrable. Let $a(\theta)$ be a continuous function such that $\int_\Theta e^{-a(\theta)}d\theta<\infty$. 

Suppose that, in case we replace the functions $h_n(\theta)$ with the functions $\tilde h_n(\theta)=h_n(\theta)-a(\theta)/n$,  these modified functions still fulfil the following condition: There exists a positive number $\gamma>0$, an open environment $U(\theta_0)$ of $\theta_0$ and a number $n_0\in\N$ such that for all $\theta\in \Theta\setminus U(\theta_0)$ we have $\tilde h_n(\theta)>\beta_0+\gamma$ $\forall n\geq n_0$. Then, \eqref{eqa:SaddlePoint} still holds. 
\end{lemma}
\begin{proof}
Note that on the left hand side of \eqref{eqa:SaddlePoint} we can replace $\pi(\theta)$ with $e^{-a(\theta)}\pi(\theta)$ and $h_n(\theta)$ with $\tilde h_n(\theta)=h_n(\theta)-a(\theta)/n$ without changing the value of the integral. Obviously, $\tilde h_n(\theta)$ also converges to $h(\theta)$ uniformly on compact sets, as the continuous function $a(\theta)$ is bounded on compact sets and thus $a(\theta)/n\to 0$ uniformly on compact sets. We can thus apply Lemma \ref{lem:SaddlePointIntegrable} to conclude.    
\end{proof}

\end{document}